%% file: paper.tex
\documentclass[a4paper,11pt]{article}
\usepackage[dvips]{color}
\usepackage[section]{algorithm}
\usepackage{algorithmic}
\usepackage{paper,latexsym,multicol,macros}
\usepackage{trbonn}
\usepackage{changebar}
\usepackage{picinpar}

\def\assumptionname{Conjecture}

\newcommand{\VAR}[1]{{\it{#1\/}}}

\def\AlgoCaption#1#2{\caption{#1\label{#2-algo}}}
\def\refalgo#1{Algorithm~\protect\ref{#1-algo}}
\def\refpagealgo#1{\pagename~\protect\pageref{#1-algo}}

\title{Exploring Grid Polygons Online}

\author{Christian Icking\thanks{%
        FernUniversit\"at Hagen, 
        Praktische Informatik VI, 
        Feithstr.~142,
        D-58084~Hagen.
}
\and    Tom Kamphans\thanks{%
        Universit\"at Bonn,
        Institut f\"ur Informatik, Abt.~I,
        R\"omerstr.~164,
        D-53117~Bonn.
}
\andnl  Rolf Klein\footnotemark[2]
\and    Elmar Langetepe\footnotemark[2]
}

\date{December 2005}
\reportnumber{001}

\begin{document}
\maketitle

{\pagestyle{empty}
\mbox{}\newpage

\begin{abstract}
We investigate the exploration problem of a short-sighted mobile robot
moving in an unknown cellular room. To explore a
cell, the robot must enter it. Once inside, the robot knows which
of the~4 adjacent cells exist and which are boundary edges.
The robot starts from a specified cell adjacent to the room's
outer wall; it visits each cell, and returns to the start. Our interest
is in a short exploration tour; that is, in keeping the number
of multiple cell visits small.
For abitrary environments containing no obstacles
we provide a strategy producing tours of length
$$S \leq C + \frac{1}{2} E - 3,$$
and for environments containing obstacles we provide a strategy, that is 
bound by
$$S \leq C + \frac{1}{2} E + 3H + \WCW - 2,$$

where $C$ denotes the number of cells---the area---, $E$ denotes
the number of boundary edges---the perimeter---, and $H$ is
the number of obstacles, and $\WCW$ is a measure for the sinuosity 
of the given environment.

\smallskip
The strategies were implemented in a Java applet \cite{hiklm-gaesu-00}
that can be found in

\centerline{\tt http://www.geometrylab.de/Gridrobot/}

\medskip

\end{abstract}

\mbox{}\newpage\mbox{}
\clearpage
}
\pagenumbering{arabic}

\input{introduction.tex}

\input{definitions.tex}

\input{complex.tex}

\newpage
\sect{Exploring Simple Polygons}{smartdfs}
\input{SmartDFSStrat.tex}
\input{SmartDFSAnalysis.tex}

\newpage
\sect{Exploring Polygons with Holes}{cellexplore}
\input{CellExplStrat.tex}
\input{CellExplAnalysis.tex}

\input{remarks.tex}

\bibliographystyle{abbrv}
\bibliography{%
        ../../../abt1/biblio/local,%
        ../../../abt1/biblio/update,%
        ../../../abt1/biblio/geom}

\end{document}

%% file: introduction.tex
\sect{Introduction}{intro}
Exploring an unknown environment and searching for a target
in unknown position are among the basic tasks of autonomous
mobile robots. Both problems have received a lot of attention
in computational geometry and in robotics; 
see, for example, \cite{ag-tsgr-03,bcr-sp-93,dkp-hlue1-98,fw-ola-98,hikk-pep-01,ik-skpcs-95,ikl-ocsws-99,iklss-ocsws-04,m-spn-97,m-gspno-00,rksi-rnut-93}.

In general, one assumes that the robot is equipped with an ideal vision
system that provides, in a continuous way, the full visibility polygon
of the robot's current position.
In practice, there are at least two reasons that limit the reach of this model.

\begin{itemize}
\item A realistic laser scanner can reliably detect objects only within a
distance of up to a few meters. Hence, in a large environment 
the robot has to move towards more distant areas to
explore them.
\item Service robots like lawn mowers or 
cleaning devices need to get close to the parts of the environment
they want to inspect and, possibly, to work on.
\end{itemize}

To accomodate these situations, we study in this paper the model of a 
rather short-sighted robot. We assume that the environment is given 
by a simple polygon, $P$. 
Its edges are either vertical or horizontal, and its vertices have integer
coordinates. Thus, the interior of $P$ consists of integer grid cells.
An interior cell is either free or blocked.%
\footnote{For convenience, we consider the exterior of the polygon to consist
of blocked cells only.}


The robot starts from a free cell, $s$, adjacent to the polygon's boundary.
From there it can enter one of the neighboring free cells, and so on.
Once inside a cell, the robot knows
which of its 4~neighbors are blocked and which are free.
The robot's task is to visit each free cell
inside~$P$ and to return to the start;
see~\reffig{figCells/RDFSExploreExample}\,(i) for an example.
This example shows a tour that visits each cell once and some cells even more
often. Our interest is in a short exploration tour, so we would like to keep
the number of excess cell visits small.

\pstexfig{(i) An example exploration tour.
          (ii) A shortest TSP tour for the same polygon.}
{figCells/RDFSExploreExample}

Even though our robot does not know its environment in advance, it is
interesting to ask how short a tour can be found in the offline situation
(i.e., when the environment is already known). This amounts to
constructing a shortest traveling salesperson tour on the free cells.

If the polygonal environment contains obstacles, the problem of finding
such a minimum length tour is known to be NP-hard, by work of Itai et
al.~\cite{ips-hpgg-82}. There are $1+\epsilon$ approximation schemes
by Grigni et al.~\cite{gkp-aspgt-95}, Arora~\cite{a-ptase-96},
and Mitchell~\cite{m-gsaps-96}, and a  
$\frac{53}{40}$ approximation by Arkin et al.~\cite{afm-aalmm-97}.

In a simple polygon without obstacles, the complexity of constructing
offline a minimum length tour seems to be open.
Ntafos~\cite{n-wrlv-92} and Arkin et al.~\cite{afm-aalmm-97}
have shown how to approximate the minimum length tour with
factors of $\frac{4}{3}$ and $\frac{6}{5}$, respectively.
Umans and Lenhart~\cite{ul-hcsgg-97} have provided an $O(C^4)$ algorithm for deciding 
if there exists a Hamiltonian cycle (i.e., a tour that visits each of the $C$ cells
of a polygon {\em exactly} once).
For the related problem of Hamiltonian paths (i.e., different
start and end positions), Everett~\cite{e-hpnrg-86} has given a polynomial algorithm
for certain grid graphs. Cho and Zelikovsky~\cite{hz-scthc-95}
studied {\em spanning closed trails} (a relaxation of Hamiltonian
cyles) in grid graphs.%
\footnote{The grid graph corresponding to a grid polygon, $P$, consists
of one node for every free cell in $P$.
Two nodes are connected by an edge, if their corresponding cells are adjacent.}


In this paper, our interest is in the online version of the cell exploration
problem. The task of exploring an grid
polygon with holes was independently considered by Gabriely and Rimon
\cite{gr-colcg-03}. They introduce a somehow artificial robot model by
distinguishing between the robot and its tool. The tool has the size of
one grid cell and moves from cell to cell. The robot moves between
the midpoints of $2\times 2$--blocks of cells.
The robot constructs a spanning tree on the $2\times 2$ cell-blocks
and the tools explores the polygon keeping the spanning tree edges on its
right side. Only if the tool's path is blocked by an obstacle, it switches
to the other side of the spanning tree and keeps the spanning tree edges on
it's left side. This model allows a smart analysis yielding an upper
bound of $C+B$, where $C$ denotes the number of cells and $B$ the number
of boundary cells.
However, this bound is larger than $C+\frac12E +3H+\WCW -2$ except for
corridors of width one, in which both bounds are the same, this may 
justify our more complicated analysis of the strategy.

Another online task is the {\em piecemeal exploration}, 
where the robot has to interrupt the exploration every 
now and then so as to return to the
start point, for example, to refuel.
Piecemeal exploration of grid graphs was studied by
Betke et~al.~\cite{brs-plue-95} and Albers et~al.~\cite{aks-eueo-02}.
Note that their objective is to visit every node {\em and} every edge,
whereas we require a complete coverage of only the cells.
Subdividing the robot's environment into grid cells is used also
in the robotics community, see, for example,
Moravec and Elfes~\cite{me-hrmwa-85}, and
Elfes~\cite{e-uogmr-89}.

Our paper is organized as follows:
In \refsect{defi} we give more detailled description of our robot and the
environment. We analyze the competitive complexity in \refsect{compcomp},
showing lower bounds for both simple polygons and polygons with holes.
In \refsect{smartdfs} we present an exploration strategy for simple 
polygons together with the analysis of the strategy and in
\refsect{cellexplore} we present and analyze a strategy for polygons
with holes.

%% file: definitions.tex
\sect{Definitions}{defi}

\pstexfig{(i) Polygon with 23 cells, 38 edges and one(!) hole (black cells), 
(ii) the robot can determine which of the 4 adjacent cells are
free, and enter an adjacent free cell.}
{figCells/CellPoly}

\begin{defi}{gridpoly}
A \DEFI{cell}~ is a basic block in our environment, defined by a tuple
$(x,y)\in \N^2$. A cell is either \DEFI{free} and can be visited by the
robot, or \DEFI{blocked}  (\IE, unaccessible for the robot).%
\footnote{In the following, we sometimes use the terms 
{\em free cells} and {\em cells} synonymously.}
We call two cells $c_1=(x_1,y_1), c_2=(x_2,y_2)$
\DEFI{adjacent} or \DEFI{neighboring}, if they share a common edge
(\IE, if $|x_1-x_2|+|y_1-y_2|=1$ holds),
and \DEFI{touching}, if they share a common edge or corner.

A \DEFI{path}, $\pfad$, from a cell $s$ to a cell $t$ is 
a sequence of free cells $s=c_1, \ldots , c_n=t$ 
where $c_i$ and $c_{i+1}$
are adjacent for $i=1,\ldots,n-1$.
Let $\len{\pfad}$ denote the length of $\pfad$.
We assume that the cells have unit size, so the length of the path
is equal to the number of \DEFI{steps} from cell to cell that the robot walks.

A \DEFI{grid polygon}, $P$, is a connected set of free cells;
that is, for every $c_1, c_2\in P$ exists a path from $c_1$ to $c_2$
that lies completely in $P$.

We call a set of touching blocked cells that are
completely surrounded by free cells an \DEFI{obstacle}
or \DEFI{hole}, see \reffig{figCells/CellPoly}.
Polygons without holes are called \DEFI{simple polygon}\DEFI{s}.
\end{defi}

\pstexfig{The perimeter, $E$, is used to distinguish between {\em thin} and
{\em thick} environments.}{figCells/fleshyskinny}

We analyze the performance of an exploration strategy using
some parameters of the grid polygon.
In addition to the area, $C$, of a polygon we introduce 
the {\em perimeter}, $E$. $C$ is the number of free cells and
$E$ is the total number of edges that appear between a free cell and a blocked
cell, see, for example, \reffig{figCells/CellPoly} 
or \reffig{figCells/fleshyskinny}.
We use $E$ to distinguish between thin and thick environments,
see \refsect{compcomp}.
In \refsect{CellExplAnalyse} we introduce another parameter, the
sinuosity $\WCW$, to distinguish between straight and twisted polygons.

%% file: complex.tex
\sect{Competitive Complexity}{compcomp}
We are interested in an online exploration. In this setting, the environment is
not known to the robot in advance.
Thus, the first question is whether the robot is still able to approximate 
the optimum solution up to a constant factor in this setting.
There is a quick and rather simple answer to this question:

\begin{theo}{ExplCompl}
The competitive complexity of exploring an unknown cellular environment with
obstacles is equal to 2.
\end{theo}
\begin{proof}
Even if the environment is unknown we can apply a simple depth-first search 
algorithm (DFS) to the grid graph.
This results in a complete exploration in $2C-2$ steps.
The shortest tour needs at least $C$ steps to visit all cells
and to return to $s$, so
DFS is competitive with a factor of 2.

\smallskip
On the other hand, 2 is also a lower bound for the 
competitive factor of any strategy.
To prove this, we construct a special grid polygon depending on the 
behavior of the strategy. 
The start position, $s$, is located in a long corridor of width~1.
We fix a large number, $Q$, and observe how the strategy explores 
this corridor.
Two cases occur.

{\bf Case 1:}
The robot eventually returns to $s$ after walking at least $Q$ 
and at most
$2Q$ steps. At this time, we close the corridor with two unvisited cells,
one at each end, see \reffig{figCells/lowbound}(i).
Let $R$ be the number of cells visited so far.
The robot has already walked at least $2R-2$ steps and needs another $2R$
steps to visit the two remaining cells and to return to~$s$, whereas the shortest
tour needs only $2R$ steps to accomplish this task.

\pstexfig{A lower bound of 2 for the exploration of grid polygons.}
        {figCells/lowbound}

{\bf Case 2:}
In the remaining case the robot concentrates---more or less---on 
one end of the corridor.
Let $R$ be the number of cells visited after $2Q$ steps.
Now, we add a bifurcation at a cell $b$ immediately behind the
farthest visited cell in the corridor, see \reffig{figCells/lowbound}(ii).
Two paths arise, which turn back and run parallel to the long corridor.
If the robots returns to $s$ before exploring one of the two paths
an argument analogous to case 1 applies. Otherwise, one of the
two paths will eventually be explored up to the cell $e$ where
it turns out that this corridor is connected to the other end 
of the first corridor.
At this time, the other path is defined to be a dead end of length $R'$,
which closes just one cell behind the last visited cell $e'$.

From $e$ the robot still has to walk to the other end of the corridor, to
visit the dead end, and to return to~$s$.
Altogether, it will have walked at least four times the 
length of the corridor, $R$,
plus four times the length of the dead end, $R'$. The optimal path needs only
$2R+2R'$, apart from a constant number of steps for the vertical 
segments.
\smallskip

In any case, the lower bound for the number of steps 
tends to $2$ while $Q$ goes to infinity.
\end{proof}

\pstexfig{A lower bound for the exploration of simple polygons.
The dashed lines show the optimal solution.}
{figCells/sdfslowerbound}

\bigskip\noindent
We cannot apply \reftheo{ExplCompl} to simple polygons, because
we used a polygon with a hole to show the lower bound. 
The following lower bound holds for simple polygons.

\begin{theo}{ExplSimplCompl}
Every strategy for the exploration of a simple grid polygon with $C$ 
cells needs at least ${7\over 6}\, C$ steps.
\end{theo}
\begin{proof}
We assume that the robot starts in a corner of the polygon, 
see \reffig{figCells/sdfslowerbound}(i)
where $\triangle$ denotes the robot's position.
Let us assume, the strategy decides to walk one step to the east---if 
the strategy walks to the south we use a mirrored construction.
For the second step, the strategy has two possibilities: Either it
leaves the wall with a step to the south, 
see \reffig{figCells/sdfslowerbound}(ii),
or it continues to follow the wall with a further step to the east,
see \reffig{figCells/sdfslowerbound}(iii).
In the first case, we close the polygon as shown in 
\reffig{figCells/sdfslowerbound}(iv).
The robot needs at least 8 steps to explore this polygon, but the 
optimal strategy needs only 6 steps yielding a factor of
${8\over 6}\approx 1.3$. In the second case we proceed as follows.
If the robot leaves the boundary, we close the polygon as shown
in \reffig{figCells/sdfslowerbound}(v) and (vi). The robot needs 12 step,
but 10 steps are sufficient.
In the most interesting case, the robot still follows the wall,
see \reffig{figCells/sdfslowerbound}(vii). In this case, the robot will need
at least 28 steps to explore this polygon, whereas an optimal strategy 
needs only 24 steps. This leaves us with a factor of 
${28\over 24}={7\over 6}\approx 1.16$.

We can easily extend this pattern to build polygons of
arbitrary size by repeating the preceding construction several times
using the {\em entry} and {\em exit} cells denoted by the arrows
in \reffig{figCells/sdfslowerbound}(iv)--(vii).
As soon as the robot leaves one block, it enters the start cell of the next
block and the game starts again; that is,
we build the next block depending on the robot's behavior.
Note that this construction cannot lead to overlapping polygons or
polygons with holes, because the polygon always extends to the
same direction.
\end{proof}

\pstexfig{DFS is not the best possible strategy.}{figCells/dfsoptimal}

Even though we have seen in \reftheo{ExplCompl} that the simple DFS strategy
already achieves the optimal competitive factor in polygons with holes, 
DFS is not the best possible exploration strategy!
There is no reason to visit {\em each} cell twice just because 
this is required in some special situations like 
dead ends of width~1. Instead, a strategy should make use
of wider areas, see \reffig{figCells/dfsoptimal}.

We use the perimeter, $E$, to distinguish between thin environments
that have many corridors of width 1, and thick environments that
have wider areas, see \reffig{figCells/fleshyskinny} 
on \refpagefig{figCells/fleshyskinny}.
In the following sections we present strategies
that explore grid polygons using no more than
roughly $C+\frac12 E$ steps.
Since all cells in the environment have to be visited, $C$ is 
a lower bound on the
number of steps that are needed to explore the whole polygon and
to return to $s$.%
\footnote{More precisely, we need at least $C-1$ steps to visit every cell, 
and at least 1 step to return to $s$.}
Thus, $\approx \frac12 E$
is an upper bound for the number of additional cell visits.
For thick environments, the value of $E$ is in $O(\sqrt{C})$, so that 
the number of additional cell visits is
substantially smaller than the number of free cells. Only for polygons
that do not contain any $2\times2$ square of free cells,
$E$ achieves its maximum value of $2(C+1)$,
and our upper bound is equal to $2C-2$, which is the cost of applying DFS. 
But in this case one cannot do better, because even the optimal
offline strategy needs that number of steps.
In other cases, our strategies are more efficient than DFS.

%% file: SmartDFSStrat.tex
We have seen in the previous section that a simple DFS traversal 
achieves a competitive factor of 2. Because the lower bound for
simple grid polygons is substantially smaller, there may be a strategy
that yields a better factor. Indeed, we can improve the DFS strategy.
In this section, we give a precise description of DFS and present two
improvements that lead to a $\frac43$-competitive exploration strategy
for simple polygons.

\ssect{An Exploration Strategy}{SmartDFSstrat}

There are four possible directions---north, south, east and west---for 
the robot to move from one cell to an adjacent cell. 
We use the command \Befehl{move({\em dir})} to execute the actual motion of
the robot.
The function \Befehl{un\-ex\-plored({\em dir})} returns
true, if the cell in the given direction seen from the robot's 
current position is not yet visited, and false otherwise.
For a given direction {\em dir}, 
\Befehl{cw({\em dir})} denotes the direction turned $90^\circ$ clockwise, 
\Befehl{ccw({\em dir})} the direction turned $90^\circ$ counterclockwise, and
\Befehl{reverse({\em dir})} the direction turned by $180^\circ$. 

Using these basic commands, 
the simple DFS strategy can be implemented as shown in \refalgo{dfs}.
For every cell that is entered in direction {\em dir}, the robot
tries to visit the adjacent cells in clockwise order, see the procedure 
{\em ExploreCell}.
If the adjacent cell is still unexplored, the robot enters this cell,
recursively calls {\em ExploreCell}, and walks back, see the procedure
{\em ExploreStep}. 
Altogether, the polygon is explored following the {\em left-hand rule}:
The robot proceeds from one unexplored cell to the next while the 
polygon's boundary or the explored cells are always to its left hand side.

\begin{algorithm}
\AlgoCaption{DFS}{dfs}
\begin{description}
\item[DFS(\VAR{P}, \VAR{start}):]~
\begin{algorithmic}
\STATE Choose direction \VAR{dir}, so that reverse(\VAR{dir})
points to a blocked cell;
\STATE ExploreCell(\VAR{dir});
\end{algorithmic}

\item[ExploreCell(\VAR{dir}):]~
\begin{algorithmic}
  \STATE \COMMENT{Left-Hand Rule:}
  \STATE ExploreStep(ccw(\VAR{dir}));
  \STATE ExploreStep(\VAR{dir});
  \STATE ExploreStep(cw(\VAR{dir}));
\end{algorithmic}

\item[ExploreStep(\VAR{dir}):]~
\begin{algorithmic}
\IF{unexplored(\VAR{dir})}
  \STATE move(\VAR{dir});
  \STATE ExploreCell(\VAR{dir});
  \STATE move(reverse(\VAR{dir}));
\ENDIF
\end{algorithmic}
\end{description}
\end{algorithm}

Obviously, all cells are visited, because the graph is connected, and
the whole path consists of $2C-2$ steps, because each cell---except
for the start---is entered exactly once by the first \Befehl{move} statement,
and left exactly once by the second \Befehl{move} statement
in the procedure {\em ExploreStep}.

\pstexfig{First improvement to DFS: Return directly to those cells that
still have unexplored neighbors.}{figCells/dfsverbesserung1}

The first improvement to the simple DFS is to return directly to those
cells that have unexplored neighbors. 
See, for example, \reffig{figCells/dfsverbesserung1}:
After the robot has reached the cell $c_1$, DFS walks to $c_2$ through the
completely explored corridor of width 2. A more efficient return path
walks on a shortest path from $c_1$ to $c_2$. Note that the robot
can use for this shortest path only cells that are already known.
With this modification, the robot's position
might change between two calls of {\em ExploreStep}. Therefore, the procedure 
{\em ExploreCell} has to store the current position, and the robot has to
walk on the shortest path to this cell, see the procedure {\em ExploreStep}
in \refalgo{smartdfs}.
The function \Befehl{unexplored({\em cell}, {\em dir})} returns true, if the
cell in direction {\em dir} from {\em cell} is not yet visited.

\pstexfig{Second improvement to DFS: Detect polygon splits.}
{figCells/dfsverbesserung2}

Now, observe the polygon shown in \reffig{figCells/dfsverbesserung2}. 
DFS completely surrounds the polygon, returns to $c_2$ and explores
the left part of the polygon. After this, it walks to $c_1$ and
explores the right part. Altogether, the robot walks four times
through the narrow corridor. A more clever solution would
explore the right part immediately after the first visit of $c_1$, and
continue with the left part after this. This solution would walk
only two times through the corridor in the middle! 
The cell $c_1$ has the property that the graph of unvisited cells
splits into two components after $c_1$ is explored. We call cells
like this \DEFI{split cell}\DEFI{s}. 
The second improvement to DFS is to recognize split cells and 
diverge from the left-hand rule when a split cell is detected.
Essentially, we want to split the set of cells into several
components, which are finished in the reversed order
of their distances to the start cell.
The detection and handling of split cells is specified in 
\refsect{analysis}.
\refalgo{smartdfs} resumes both improvements to DFS.

\pstexfig{Straightforward strategies are not better than SmartDFS.}
{figCells/SimpleStrats}

Note that the straightforward strategy {\em Visit all boundary cells and
calculate the optimal offline path for the rest of the polygon}
does not achieve a competitive factor better than $2$. For example, in 
\reffig{figCells/SimpleStrats}(i) this strategy visits almost every 
boundary cell twice, whereas SmartDFS visits only one cell twice.
Even if we extend the simple strategy to detect split cells while visiting
the boundary cells, we can not
achieve a factor better than $\frac43$. A lower bound on the
performace of this strategy is a corridor of width 3, see
\reffig{figCells/SimpleStrats}(ii).
Moreover, it is not known whether the offline strategy is NP-hard for
simple polygons.

\begin{algorithm}
\AlgoCaption{SmartDFS}{smartdfs}
\begin{description}
\item[SmartDFS(\VAR{P}, \VAR{start}):]~
\begin{algorithmic}
\STATE Choose direction \VAR{dir} for the robot, 
so that reverse(\VAR{dir}) points to 
\STATE ~~~a blocked cell;
\STATE ExploreCell(\VAR{dir});
\STATE Walk on the shortest path to the start cell;
\end{algorithmic}

\item[ExploreCell(\VAR{dir}):]~
\begin{algorithmic}
\STATE Mark the current cell with the number of the current layer;
\STATE \VAR{base} $:=$ {\rm current position};
\IF{not isSplitCell(\VAR{base})}
  \STATE \COMMENT{Left-Hand Rule:}
  \STATE ExploreStep(\VAR{base}, ccw(\VAR{dir}));
  \STATE ExploreStep(\VAR{base}, \VAR{dir});
  \STATE ExploreStep(\VAR{base}, cw(\VAR{dir}));
\ELSE
  \STATE \COMMENT{choose different order, see \refpage{comporder}}\,ff
  \STATE Determine the types of the components using the layer numbers
  \STATE ~~~of the surrounding cells;
  \IF { No component of type III exists }
     \STATE Use the left-hand rule, but omit the first possible step.
  \ELSE
     \STATE Visit the component of type III at last.
  \ENDIF
\ENDIF
\end{algorithmic}

\item[ExploreStep(\VAR{base}, \VAR{dir}):]~
\begin{algorithmic}
\IF{unexplored(\VAR{base}, \VAR{dir})}
  \STATE Walk on shortest path using known cells to $base$;
  \STATE move(\VAR{dir});
  \STATE ExploreCell(\VAR{dir});
\ENDIF
\end{algorithmic}
\end{description}
\end{algorithm}


%% file: SmartDFSAnalysis.tex
\ssect{The Analysis of SmartDFS}{analysis}
SmartDFS explores the polygon in layers: Beginning with the cells along the
boundary, SmartDFS proceeds towards the interior of $P$. 
Let us number the single layers:

\begin{defi}{layers} Let $P$ be a (simple) grid polygon. The 
boundary cells of $P$ uniquely define the {\em first layer} of $P$. 
The polygon 
$P$ without its first layer is called the {\em 1-offset} of $P$. 
The \ith{$\ell$} layer and the $\ell$-offset of $P$ are defined 
successively, see \reffig{figCells/doffset}.
\end{defi}

\pstexfig{The 2-offset (shaded) of a grid polygon $P$.}{figCells/doffset}

Note that the $\ell$-offset of a polygon $P$ is not necessarily
connected. Although the preceding definition is independent from any
strategy, SmartDFS can determine a cell's layer when the cell is visited
for the first time. We can define the $\ell$-offset in the same
way for a polygon with holes, but the layer of a given cell can no longer 
be determined on the first visit in this case. 
The $\ell$-offset has an important property:

\begin{lem}{offset}
The $\ell$-offset of a simple grid 
polygon, $P$, has at least $8\ell$ edges fewer than $P$.
\end{lem} 

\begin{proof}
First, we can cut off blind alleys that are narrower than $2\ell$, because
those parts of $P$ do not affect the $\ell$-offset.
We walk clockwise around the boundary cells of the remaining polygon, see
\reffig{figCells/doffset}. For every left turn the offset gains at most $2\ell$ 
edges and for every right turn the offset looses at least $2\ell$ edges. 
O'Rourke \cite{o-agta-87} showed that 
$\mbox{\#vertices} = 2\cdot \mbox{\#reflex vertices} +4$ holds for
orthogonal polygons, so
there are four more right turns than left turns. 
\end{proof}

\pstexfig{A decomposition of $P$ at the split cell $c$ and 
its handling in {\rm SmartDFS}.}{figCells/Decomposition}

\refdefi{layers} allows us to specify the detection and handling of
a split cell in \smartDFS. We start with the handling of a
split cell and defer split cell detection.

%
%
Let us consider the situation shown in \reffig{figCells/Decomposition}(i)
to explain the handling of a split cell.
\smartDFS\ has just met the
first split cell, $c$, in the fourth layer of $P$. $P$ divides into
three parts:
$$P = K_1 \disjoint K_2 \disjoint\ \{\,\mbox{visited cells of } P\,\},$$
where $K_1$ and $K_2$ denote the connected components of the set of unvisited
cells.
In this case it is reasonable to explore the component $K_2$ first, because
the start cell $s$ is closer to $K_1$; that is, we can extend
$K_1$ with $\ell$ layers, such that the resulting polygon contains 
the start cell $s$.

%
%
More generally, we want to divide our polygon $P$ into two parts, $P_1$
and $P_2$, so that each of them is an extension of the two
components. Both polygons overlap in the area around the split cell $c$.
At least one of these polygons contains the start cell. If only one
of the polygons contains $s$, we want our strategy to 
explore this part at last, expecting that in this part the path from the 
last visited cell
back to $s$ is the shorter than in the other part.
Vice versa, if there is a polygon that does {\em not} contain $s$, we explore 
the corresponding component first. In \reffig{figCells/Decomposition},
\smartDFS\ recursively enters $K_2$, returns to 
the split cell $c$, and explores the component $K_1$ next.

%
%
In the preceding example, there is only one split cell in $P$, but
in general there will be a sequence of split cells, $c_1,\ldots,c_k$.
In this case, we apply the handling of split cells in a recursive way;
that is, if a split cell $c_{i+1}, 1\leq i< k$, is detected in one of the two 
components occurring at $c_i$ we proceed the same way as described earlier. 
Only the role of the start cell is now played by the preceding 
split cell $c_i$. In the following, the term {\em start cell}\/ always
refers to the start cell of the current component; that is, either
to $s$ or to the previously detected split cell. 
Further, it may occur that three components arise at a split cell, see
\reffig{figCells/splitcell2}(i) on \refpagefig{figCells/splitcell2}.
We handle this case as two successive polygon splits
occurring at the same split cell.

%
%

\pstexfig{Several types of components.}{figCells/splitcell}

\newpage

\paragraph*{Visiting Order}~\\
We use the layer numbers to decide which component we have 
to visit at last. Whenever a split cell occurs in layer $\ell$,
every component is one of the following\labelpage{comporder}
types, see \reffig{figCells/splitcell}:

\begin{romanlist}
\item[I.] $K_i$ is {\em completely} surrounded by layer $\ell$%
\footnote{More precisely, the part of layer $\ell$ that surrounds $K_i$ is
completely visited. For convenience, we use the slightly sloppy, but shorter 
form.}
\labelpage{comptypes}
\item[II.] $K_i$ is {\em not} surrounded by layer $\ell$
\item[III.] $K_i$ is {\em partially} surrounded by layer $\ell$
\end{romanlist}

\pstexfig{Switching the current layer.}{figCells/layerswitch}

\smallskip
There are two cases, in which \smartDFS\ switches from a layer $\ell-1$ 
to layer $\ell$. Either it reaches the first cell of
layer $\ell-1$ in the current component and thus passes the start 
cell---see, for example, the switch from layer 1 to layer 2 in 
\reffig{figCells/layerswitch}---, or it hits another cell of 
layer $\ell-1$ but no polygon split occurs, such as the 
switch from layer 2 to layer 3 in in \reffig{figCells/layerswitch}.
In the second case, the considered start cell must be located in a narrow 
passage that is completely explored; otherwise, the strategy would
be able to reach the first cell of layer $\ell-1$ as in the first case.
In both cases the part of $P$ surrounding a component of type III 
contains the first cell of the current layer $\ell$
as well as the start cell. 
Therefore, it is reasonable to explore the component of type III at last.

There are two cases, in which no component of type III exists 
when a split cell is detected:
\begin{enumerate}
\item The part of the polygon that contains the preceding start cell is
  explored completely, see for example \reffig{figCells/splitcell2}(i).
  In this case the order of the components makes no difference.%
  \footnote{In \reffig{figCells/splitcell2}(i) we gain two steps, if we
    explore the part left to the splitcell at last and 
    do not return to the split cell after this part is completely
    explored, but return immediately to the start cell. But
    decisions like this require facts of much more global type 
    than we consider up to now.
    However, for the analysis of our strategy and the upper bound
    shortcuts like this do not matter.}

\item Both components are completely surrounded by a layer,
  because the polygon split and the switch from one layer to the
  next occurs within the same cell, see \reffig{figCells/splitcell2}(ii).
  A step that follows the left-hand rule will move towards the
  start cell, so we just omit this step. More precisely, if the
  the robot can walk to the left, we prefer a step forward
  to a step to the right. If the robot cannot walk to the left
  but straight forward, we proceed with a step to the right.
\end{enumerate}

We proceed
with the rule in case 2 whenever there is no component of type III,
because the order in case 1 does not make a difference.

\pstexfig{No component of type III exists.}{figCells/splitcell2}

%
%
\newpage 
\paragraph*{An Upper Bound on the Number of Steps}~\\
For the analysis of our strategy we consider two 
polygons, $P_1$ and $P_2$, as follows.
Let $Q$ be the square of width $2q+1$ around $c$ with
$$q:=\cases{\ell, & if $K_2$ is of type I\cr \ell-1, 
                  & if $K_2$ is of type II},$$
where $K_2$ denotes the component that is explored first, and
$\ell$ denotes the layer in which the split cell was found.
We choose $P_2 \subset P\cup Q$ such that $K_2 \cup \{c\}$ 
is the $q$-offset of $P_2$, 
and $P_1 := ((P\backslash P_2) \cup Q) \cap P$,
see \reffig{figCells/Decomposition}. 
The intersection with $P$ is necessary, because $Q$ may exceed the boundary
of $P$. Note that at least $P_1$ contains the preceding start cell.
There is an arbitrary number of polygons $P_2$, such that
$K_2 \cup \{c\}$ is the $q$-offset of $P_2$, because blind alleys of $P_2$ 
that are not wider than $2q$ do not affect the $q$-offset.
To ensure a unique choice of $P_1$ and $P_2$, we require
that both $P_1$ and $P_2$ are connected, and both $P \cup Q=P_1 \cup P_2$ 
and $P_1 \cap P_2 \subseteq Q$ are satisfied.

The choice of $P_1, P_2$ and $Q$ ensures that the robot's path
in $P_1\backslash Q$ and in $P_2\backslash Q$ do not change 
compared to the path in $P$. The parts of the robot's path
that lead from $P_1$ to $P_2$ and from $P_2$ to $P_1$
are fully contained in the square $Q$.
Just the parts inside $Q$ are bended to connect the appropriate
paths inside $P_1$ and $P_2$,
see \reffig{figCells/Decomposition} and \reffig{figCells/Decomposition2}.

\pstexfig{The component $K_2$ is of type I. The square $Q$ may exceed
$P$.}{figCells/Decomposition2}

In \reffig{figCells/Decomposition}, $K_1$ is of type III and $K_2$ is 
of type II. A component of type I occurs, if we detect a split cell
as shown in \reffig{figCells/Decomposition2}.
Note that $Q$ may exceed $P$, but $P_1$ and $P_2$
are still well-defined.

%
%

\pstexfig{The order of components is not necessarily optimal.}
{figCells/SplitCellNotOptimal}

Remark that we do not guarantee that the path from the last 
visited cell back to the corresponding start cell is the
shortest possible path. See, for example,
\reffig{figCells/SplitCellNotOptimal}:
A split cell is met in layer 2. Following the preceding rule, 
\smartDFS\ enters $K_2$ first, returns to $c$, explores $K_1$,
and returns to $s$. A path that visits $K_1$ first and moves
from the upper cell in $K_2$ to $s$ is
slightly shorter. A case like this may occur if the first cell of the current 
layer lies in $Q$.
However, we guarantee that there is only one return path 
in $P_1\backslash Q$ and in $P_2\backslash Q$; that is, only
one path leads from the last visited cell back to the preceding
start cell causing double visits of cells.

\bigskip

%
%
We want to visit every cell in the polygon and to return to
$s$. Every  strategy needs at least $C(P)$ steps to fulfill this task,
where $C(P)$ denotes the number of cells in $P$.
Thus, we can split the overall length of the exploration path, $\pfad$,
into two parts, $C(P)$ and $\excess(P)$, with $\len{\pfad}=C(P)+\excess(P)$.
$C(P)$ is a lower bound on the number of steps that are needed for
the exploration task, whereas
$\excess(P)$ is the number of additional cell visits.

Because \smartDFS\ recursively explores $K_2\cup\{c\}$, we want to
apply the upper bound inductively to the component $K_2\cup \{c\}$. 
If we explore $P_1$ with \smartDFS\ until $c$ is met, 
the set of unvisited cells of $P_1$ is equal to $K_1$, because the 
path outside $Q$ do not change. Thus, we can apply
our bound inductively to $P_1$, too. 
The following lemma gives us the relation between
the path lengths in $P$ and the path lengths in the two components.

\begin{lem}{component}
Let $P$ be a simple grid polygon. 
Let the robot visit the first 
split cell, $c$, which splits the unvisited cells of $P$ into two components 
$K_1$ and $K_2$, where $K_2$ is of type I or II. 
With the preceding notations we have
$$\excess(P)\leq \excess(P_1)+ \excess(K_2\cup\{c\})+1\; .$$
\end{lem}
\begin{proof}
The strategy \smartDFS\ has reached the split cell $c$ and 
explores $K_2\cup\{c\}$ with start cell $c$ first. Because $c$ is the
first split cell, there is no excess in $P_2\backslash (K_2\cup\{c\})$
and it suffices to consider $\excess(K_2\cup\{c\})$ for this
part of the polygon. 
After $K_2\cup\{c\}$ is finished, the robot returns to $c$ 
and explores $K_1$. For this part we take 
$\excess(P_1)$ into account. 
Finally, we add one single step, because the split cell $c$
is visited twice: once, when \smartDFS\ detects the split and
once more after the exploration of $K_2\cup\{c\}$ is
finished. Altogether, the given bound is achieved. 
\end{proof}

$c$ is the first split cell in $P$, so $K_2\cup\{c\}$ is the $q$-offset of
$P_2$ and we can apply \reflem{offset} to bound
the number of boundary edges of $K_2\cup\{c\}$ by the number
of boundary edges of $P_2$. The following lemma
allows us to charge the number of edges in $P_1$ and $P_2$
against the number of edges in $P$ and $Q$.

\begin{lem}{edgecount}
Let $P$ be a simple grid polygon, and let $P_1, P_2$ and $Q$ be defined as 
earlier. The number of edges satisfy the equation
$$E(P_1) + E(P_2) = E(P) + E(Q)\;.$$
\end{lem}
\begin{proof}
Obviously, two arbitrary polygons $P_1$ and $P_2$ always satisfy
$$E(P_1) + E(P_2) = E(P_1\cup P_2) + E(P_1\cap P_2)\;.$$

Let $Q' := P_1 \cap P_2$. Note that $Q'$ is not necessarily the same
as $Q$, see, for example, \reffig{figCells/Decomposition2}. With 
$P_1\cup P_2=P\cup Q$ we have
\begin{eqnarray*}
E(P_1) + E(P_2) & = & E(P_1\cap P_2) + E(P_1 \cup P_2)\\
& = & E(Q') + E(P\cup Q)\\
& = & E(Q') + E(P) + E(Q) - E(P\cap Q)\\
& = & E(P) + E(Q)
\end{eqnarray*}
The latter equation holds because $Q'=P\cap Q$.
\end{proof}

Finally, we need an upper bound for the length of a path inside a grid polygon.

\begin{lem}{shortest}
Let $\pfad$ be the shortest path between two cells 
in a grid polygon $P$.
The length of $\pfad$ is bounded by
$$\len{\pfad} \leq \frac12 E(P) -2\;.$$
\end{lem}
\begin{proof}
\Wlog\ we can assume that the start cell, $s$, and the target cell, $t$, of 
$\pfad$ belong to the first layer of $P$,
because we are searching for 
an upper bound for the shortest path between two arbitrary cells. 

Observe the path $\pfad_L$ from $s$ to $t$ in the first layer that follows
the boundary of $P$ clockwise and the path $\pfad_R$ that follows
the boundary counterclockwise. The number of edges along 
these paths is at least four greater than the number of cells
visited by $\pfad_L$ and $\pfad_R$ using an argument similar to the
proof of \reflem{offset}. Therefore we have:
$$\len{\pfad_L} + \len{\pfad_R} \leq E(P)-4.$$

In the worst case, both paths have the same length, so
$\len{\pfad(s,t)} = \len{\pfad_L} = \len{\pfad_R}$ holds. With
this we have
$$2\cdot\len{\pfad(s,t)} \leq E(P)-4 \;\Longrightarrow\;  
\len{\pfad(s,t)} \leq \frac12 E(P)-2.$$
\end{proof}

\noindent
Now, we are able to show our main theorem:
\begin{theo}{SDFSmain}
Let $P$ be a simple grid polygon with $C$ cells and $E$ edges. $P$ can
be explored with
$$S\leq C+\frac12 E-3$$
steps. This bound is tight.
\end{theo}
\begin{proof}
$C$ is the number of cells and thus a lower bound on the number of steps that
are needed to explore the polygon $P$. We show by an induction 
on the number of components 
that $\excess(P) \leq \frac12 E(P)-3$ holds.

For the induction base we consider a polygon without any split cell:
\smartDFS\ visits each cell and returns on the shortest path 
to the start cell. Because there is no polygon split, all cells of $P$ 
can be visited by a path of length $C-1$. 
By \reflem{shortest} the shortest path back to the start cell
is not longer than $\frac12 E -2$; thus,
$\excess(P)\leq \frac12 E(P)-3$ holds.

Now, we assume that there is more than one component during the 
application of \smartDFS. 
Let $c$ be the first split cell detected in $P$. 
When \smartDFS\ reaches $c$, two new components, $K_1$ and $K_2$, occur.  
We consider the two polygons $P_1$ and $P_2$ defined as earlier,
using the square $Q$ around $c$.

\Wlog\ we assume that $K_2$ is recursively explored first with
$c$ as start cell.
After $K_2$ is completely explored, \smartDFS\ proceeds with 
the remaining polygon.
As shown in \reflem{component} we have
$$\excess(P)\leq \excess(P_1)+ \excess(K_2\cup\{c\})+1\; .$$
Now, we apply the induction hypothesis to $P_1$ and $K_2\cup\{c\}$
and get
$$\excess(P)\leq \frac12 E(P_1)-3+\frac12 E(K_2\cup\{c\})-3 +1\; .$$
By applying \reflem{offset} to the $q$-offset $K_2\cup\{c\}$ of $P_2$ 
we achieve 
\begin{eqnarray*}
\excess(P) 
& \leq & \frac12 E(P_1)-3+\frac12(E(P_2)-8q)-3 +1\\
& = &  \frac12( E(P_1)+E(P_2) ) -4q -5 \; .
\end{eqnarray*}
From \reflem{edgecount} we conclude $E(P_1)+E(P_2)\leq E(P) + 4(2q+1)$.
Thus, we get $\excess(P)\leq  \frac12 E(P)-3$.

In \refsect{compcomp} we have already seen that the bound is 
exactly achieved in polygons that do not contain any $2\times 2$-square
of free cells.
\end{proof}  

\paragraph*{Competitive Factor}~\\
So far we have shown an upper bound on the number of steps needed
to explore a polygon that depends on the number of cells and
edges in the polygon. Now, we want to analyze \smartDFS\ 
in the competitive framework.

Corridors of width 1 or 2 play a crucial role in the following,
so we refer to them as {\em narrow passages}. More precisely, 
a cell, $c$, belongs to a narrow passage, if $c$ can be removed without
changing the layer number of any other cell. 

It is easy to see that narrow passages are explored optimally:
In corridors of width 1 both \smartDFS\ and the optimal strategy
visit every cell twice, and in the other case both strategies visit
every cell exactly once. 

We need two lemmata to show a competitive factor for \smartDFS.
The first one gives us a relation between the number of cells 
and the number of edges for a special class of polygons.

\begin{lem}{NoNarrow}
For a simple grid polygon, $P$, with $C(P)$ cells and $E(P)$ edges, and without
any narrow passage or split cells in the first layer, we have
$$E(P) \leq \frac23\,C(P)+6\,.$$
\end{lem}
\begin{proof}
Consider a simple polygon, $P$. We successively remove a row or column
of at least three boundary cells, maintaining our assumption that
the polygon has no narrow passages or split cells in the first layer.
These assumptions ensure that we can always find such a row or column
(i.e., if we cannot find such a row or column, the polygon has a narrow
passage or a split cell in the first layer).
Thus, we remove at least three cells and at most two edges. This decomposition
ends with a $3\times 3$ block of cells that fulfills $E=\frac23C+6$.
Now, we reverse our decomposition; that is, we successively add all rows and
columns until we end up with $P$. In every step, we add at least three
cells and at most two edges. Thus, $E \leq \frac23C+6$ is fulfilled in
every step.
\end{proof}

\pstexfig{For polygons without narrow passages or split cells in the first 
layer, the last explored cell, $c'$, lies in the 1-offset, $P'$ (shaded).}
{figCells/sdfscomplayer}

For the same class of polygons, we can show that \smartDFS\ behaves
slightly better than the bound in \reftheo{SDFSmain}.

\begin{lem}{StepNoNarrow}
A simple grid polygon, $P$, with $C(P)$ cells and $E(P)$ edges, and without
any narrow passage or split cells in the first layer can be 
explored using no more steps than
$$S(P) \leq C(P) +\frac12E(P) - 5\,.$$
\end{lem}
\begin{proof}
In \reftheo{SDFSmain} we have seen that $S(P) \leq C(P) +\frac12E(P) - 3$
holds. To show this theorem, we used \reflem{shortest} on \refpagelem{shortest}
as an upper bound for
the shortest path back from the last explored cell to the start cell.
\reflem{shortest} bounds the shortest path from a cell, $c$, in the
first layer of $P$ to the cell $c'$ that maximizes the distance to 
$c$ inside $P$; thus, $c'$ is located in the first layer of $P$, too.

Because $P$ has neither narrow passages nor split cells in the first layer,
we can explore the first layer of $P$ completely before we visit another
layer, see \reffig{figCells/sdfscomplayer}. 
Therefore, the last explored cell, $c'$, of $P$ is located in the 
1-offset of $P$. Let $P'$ denote the 1-offset of $P$, and $s'$ the first
visited cell in $P'$. Remark that $s$ and $s'$ are at least touching
each other, so the length of a shortest path from $s'$ to $s$ is at 
most 2.
Now, the shortest path, $\pfad$, from $c'$ to $s$ in $P$ is bounded by 
a shortest path, $\pfad'$, from $c'$ to $s'$ in $P'$ and a
shortest path from $s'$ to $s$:
$$|\pfad| \leq |\pfad'| + 2\,.$$
The path $\pfad'$, in turn, is bounded using \reflem{shortest} by
$$|\pfad'| \leq \frac12 E(P')- 2\,.$$
By \reflem{offset} (\refpagelem{offset}), $E(P') \leq E(P) -8$ holds,
and altogether we get
$$|\pfad| \leq \frac12 E(P) - 4\,,$$
which is two steps shorter than stated in \reflem{shortest}.
\end{proof}

\medskip

\noindent
Now, we can prove the following
\begin{theo}{sdfscomp}
The strategy {\smartDFS} is $\frac43$-competitive.
\end{theo}
\begin{proof}
Let $P$ be a simple grid polygon. In the first stage, we remove
all narrow passages from $P$ and get a sequence of (sub-)polygons
$P_i$, $i=1,\ldots,k$, without narrow passages.
For every $P_i$, $i=1,\ldots,k-1$, the optimal strategy 
in $P$ explores the part of $P$ that corresponds to $P_i$ up to the
narrow passage that connects $P_i$ with $P_{i+1}$, enters $P_{i+1}$, and
fully explores every $P_j$ with $j\geq i$. Then it returns to $P_i$ and
continues with the exploration of $P_i$. Further, we already know that
narrow passages are explored optimally. This allows us to 
consider every $P_i$ separately without changing the competitive factor 
of $P$.

Now, we observe a (sub-)polygon $P_i$. We show by induction on the
number of split cells in the first layer that 
$S(P_i) \leq \frac43C(P_i) -2$ holds. 
Note that this is exactly 
achieved in polygons of 
size $3\times m$, $m$ even, see \reffig{figCells/sdfscomp}.

\pstexfig{In a corridor of width 3 and even length,
$S(P)=\frac43\, \rmsub{S}{Opt}(P)-2$ holds.}
{figCells/sdfscomp}

If $P_i$ has no split cell in the first layer (induction base), 
we can apply 
\reflem{StepNoNarrow} and \reflem{NoNarrow}:
\begin{eqnarray*}
S(P_i) &\leq & C(P_i)+\frac12\,E(P_i)-5 \\
& \leq & C(P_i)+\frac12\left(\frac23\, C(P_i)+6\right)-5\\
& = & \frac43\, C(P_i) -2\,.
\end{eqnarray*}

\pstexfig{Three cases of split cells, (i) component of type II, (ii) 
and (iii) component of type I.}
{figCells/sdfscomp2cases}

Two cases occur if we meet a split cell, $c$, in the first layer, see
\reffig{figCells/sdfscomp2cases}. In the first case, the new component
was never visited before (component of type II, see \refpage{comptypes}).
Here, we define $Q:=\{c\}$.
The second case occurs, because the robot meets a cell, $c'$, 
that is in the first layer and touches the current cell, $c$, 
see for example \reffig{figCells/sdfscomp2cases}(ii) and (iii).
Let $Q$ be the smallest rectangle that contains both $c$ and $c'$.

Similar to the proof of \reftheo{SDFSmain}, we split the polygon
$P_i$ into two parts, both including $Q$. Let $P''$ denote 
the part that includes the component of type I or II, $P'$ the
other part. For $|Q|=1$, see \reffig{figCells/sdfscomp2cases}(i),
we conclude $S(P_i) = S(P')+S(P'')$ and
$C(P_i) = C(P')+C(P'')-1$. Applying the induction hypothesis to
$P'$ and $P''$ yields
\begin{eqnarray*}
S(P_i) & = & S(P')+S(P'')\\
&\leq & \frac43\,C(P')-2 + \frac43\,C(P'')-2\\
&=&  \frac43\,C(P_i)+\frac43-4 \quad < \quad \frac43\,C(P_i)-2\,.
\end{eqnarray*}

For $|Q|\in \{\,2,4\,\}$ we gain some steps by merging the polygons.
If we consider $P'$ and $P''$ separately, we count the steps from $c'$ to
$c$---or vice versa---in both polygons, but in $P_i$ the path from $c'$
to $c$ is replaced by the exploration path in $P''$. Thus, we
have $S(P_i) = S(P')+S(P'')-|Q|$ and 
$C(P_i) = C(P')+C(P'')-|Q|$. This yields
\begin{eqnarray*}
S(P_i) & = & S(P')+S(P'')-|Q|\\
&\leq & \frac43\,C(P')-2 + \frac43\,C(P'')-2-|Q|\\
&=&  \frac43\,C(P_i)+\frac13\,(|Q|-6)-2 \quad < \quad \frac43\,C(P_i)-2\,.
\end{eqnarray*}

The optimal strategy needs at least $C$ steps, which, altogether, yields a 
competitive factor of $\frac43$. 
\end{proof}

\newpage

\paragraph*{Split Cell Detection}~\\
\pstexfig{(i) Detecting a split cell, (ii) and (iii) a polygon split occurs 
in layer 1.}
{figCells/splitDetection}

Finally, we describe the {\em detection} of a split cell. In the first 
instance, let us assume that the robot already moves in 
a layer $\ell > 1$.
In \refsect{SmartDFSstrat} we defined that a split cell divides the
graph of unexplored cells into two parts when the split cell
is visited. Because the polygon is simple, we can determine a global
split using a local criterion. We observe the eight cells
surrounding the current robot's position. If there is more than
one connected set of visited cells in this block, the
current robot position is obviously a split cell, see
\reffig{figCells/splitDetection}(i).
Remark that we can execute this test although the robot's sensors
do not allow us to access all eight cells around the current position.
We are interested only in visited cells, so we can use the
robot's map of visited cells for our test in layers $\ell > 1$.
Unfortunately, this test method fails in layer 1, because the robot
does not know the ``layer 0'', the polygon walls.
However, we want to visit the component that has no visited cell
in the current layer (type II) first; therefore, a step that
follows the left-hand rule is correct. The strategy behaves
correctly, although we do not report the splitcell explicitly.
See, for example, \reffig{figCells/splitDetection}(ii) and (iii):
In both cases the polygon split cannot be detected in $c$, because
the cell marked with '?' is not known at this time. The split
will be identified and handled correctly in $c'$.

%% file: CellExplStrat.tex
In an environment with obstacles ({\em holes}) it is not obvious how
to detect and to handle split cells.
When a polygon split is detected, the robot may be far away from the 
split cell, because it had no chance to recognize the split before reaching
its current position.
For example, in \reffig{figCells/SplitInCellExpl}(i) the robot has 
surrounded one single obstacle and $c$ is a split cell,
whereas in (ii) there are two obstacles and $c$ is no split cell.
Both situations cannot be distinguished until the cell $c'$ is
reached.
So we use a different strategy to explore environments with obstacles.

\pstexfig{In an environment with obstacles, the robot may detect a 
split on a position far away from the splitcell, (i) $c$ was a 
split cell, (ii) $c$ was no split cell.}{figCells/SplitInCellExpl}

\ssect{An Exploration Strategy}{CellExplStrat}
The basic idea of our strategy, \Strategyname,
is to {\em reserve} all cells right 
to the covered path for the way back. As in \smartDFS\ we use the
left-hand rule; that is, the robot proceeds keeping the polygon's boundary 
or the reserved cells on its left side.
\Strategyname\ uses two modes. In the
forward mode the robot enters undiscovered parts of the polygon, and
in the backward mode the robot leaves known parts, see
\refalgo{cellexplore} on \refpagealgo{cellexplore}.
We require that the robot starts with
its back to a wall. 

\newpage 

\pstexfig{Example of an exploration tour produced by CellExplore 
(Screenshot using \cite{hiklm-gaesu-00}; the white cells are holes, dark
gray cells are reserved).}
{figCells/examples}

\reffig{figCells/examples} 
shows an example of an exploration tour.
The robot starts in J1 and explores the polygon in the forward mode
until F8 is reached. There, the robot switches to the backward mode
and explores the reserved cells F7--F5. The path from F5 to H5 is
blocked by the hole in G5, so the robot walks on the cells F4--H4
which have been visited already in the forward mode. In H5 the
robot discovers the unreserved and unexplored cell H6,
switches back to the forward mode and explores the cells H6--H8.
Note that no cells can be reserved in this case, because 
the cells I6--I8 have already been reserved during the exploration
of J8--J6. Therefore, the robot walks the same path back to H5 and
continues the return path in the backward mode. In the forward mode,
the robot could not reserve a cell from H3, so we move via
H4 to I4 and proceed with the return path in the backward mode.
The cells D9, C2 and G2 are blocked, so the robot has to circumvent
these cells using visited cells. In C5 another unreserved and unexplored cell 
is discovered, so we switch to the forward mode and visit C4.

\newpage
\mbox{}

\vfill
\begin{algorithm}[h]
\AlgoCaption{CellExplore}{cellexplore}
\begin{description}
\item[Forward mode:]~
  \begin{itemize}
  \item The polygon is explored following the left-hand rule:
    For every entered cell the robot tries to extend its path to
    an adjacent, unexplored, and unreserved cell, preferring a step to
    the left\footnotemark\
    over a straight step over a step to the right.
  \item All unexplored and unreserved cells right to the covered path 
    are reserved for the return path by pushing them onto a stack.\footnotemark\
    If no cell right to the robot's current position can be reserved---because
    there is a hole or the corresponding cell is already reserved
    or explored---the robot's position is pushed onto the stack for the 
    return path.
  \item Whenever no step in the forward mode is possible, the strategy enters the
    backward mode.
  \end{itemize}

\item[Backward mode:]~
  \begin{itemize}
  \item The robot walks back along the reserved return path.
  \item Whenever an unexplored and unreserved cell appears adjacent to the
    current position, the forward mode is entered again.
  \end{itemize}
\end{description}
\end{algorithm}

\vfill
\mbox{}
 \addtocounter{footnote}{-1}
 \footnotetext{A ``step to the left'' or ``turn left'' means that the robot 
  turns $90^\circ$ counterclockwise and moves one cell forward. Analogously
  with ``step to the right'' or ``turn right''.}
 \addtocounter{footnote}{1}
 \footnotetext{If the robot turns left, we reserve three cells: right hand,
straight forward, and forward-right to the robot's position.
  Note that we store the markers only in the robot's memory. This allows
  us to reserve cells that only touch the current cell, even if we are
  not able to determine whether these cells are free or blocked.}

\pstexfig{A polygon with $C=69, \frac{E}2=52, H=1, \WCW=2,
S=124=C+\frac12E+3H+\WCW-2$.
The return path in this polygon cannot be shortened.}
{figCells/CellExplSPTight}

\bigskip

A straightforward improvement to the strategy
\Strategyname\ is to use in the backward mode
the shortest path---on the cells known so far---to the first
cell on the stack that is unexplored or has unexplored neighbors
instead of walking back using every reserved cell,
see the first improvement of DFS.
From a practical point of view, this improvement is
very reasonable, because the performance of the strategy increases in
most environments.
Unfortunately, the return path (\IE, the path walked in the backward mode)
is no longer determined by a local configuration of cells. Instead, we
need a global view, which complicates the analysis of this strategy.
However, there are polygons that
force this strategy to walk exactly the same return path
as \Strategyname\ without any optimization,
see \reffig{figCells/CellExplSPTight}, so this idea does not improve
the worst case performance, and the upper bound for the number of steps
is the same as in \reftheo{schrankefinal}.

%% file: CellExplAnalysis.tex
\ssect{The Analysis of \Strategyname}{CellExplAnalyse}
\begin{smallfig}{figCells/assumption}
We analyze \Strategyname\ in three steps: First, we
analyze the single steps of the strategy with a local view and with one
assumption concerning the robot's initial position, see the figure
on the right. 
This results in a bound that depends on the number of left turns
that are needed to explore the polygon.
Then we discard the assumption, and finally we consider some global arguments
to replace the number of left turns with parameters of the polygon.
\end{smallfig}

\noindent
To analyze our strategy \Strategyname\ we use the following observations:
\begin{itemize}
\item \Strategyname\ introduces a dissection of all cells 
  into cells that are explored in the forward mode, denoted by $\MFWMode$, 
  and cells that are explored in the backward mode, $\MBKMode$:
  $\MZellen = \MFWMode \stackrel{\bullet}{\cup} \MBKMode$.
\item All cells that are explored in the forward mode can be
  uniquely classified by the way the robot leaves them: either it 
  makes a step to
  the left, a step forward or a step to the right.
  $\MFWMode = \MFWMode_L \stackrel{\bullet}{\cup} \MFWMode_F
  \stackrel{\bullet}{\cup} \MFWMode_R.$
\item \Strategyname\ defines a mapping 
  $\varphi: \MBKMode \longrightarrow \MFWMode$, 
  which assigns to every cell $c \in \MBKMode$ one cell $d \in \MFWMode$,
  so that $c$ was reserved while $d$ was visited.
\item There exists a subset $\MSplit \subseteq \MBKMode$
  of the cells that have an unexplored and unreserved neighbor and
  the strategy switches from backward mode to forward mode.
  We call these cells \DEFI{division cells}.
\end{itemize}

\pstexfig{Decomposing a polygon. $\triangle$ denotes the start cell and 
the initial direction. $\Delta C, \Delta E$ and $\Delta S$ denote the
differences in the number of cells, edges, and steps, respectively. $G$ 
denotes the balance.}
{figCells/knabbernbeispiel2}

\clearpage

\pstexfig{Decomposing a polygon. The shaded part shows the reserved cells.}
{figCells/knabbernbeispiel}

We will analyze \Strategyname\ by an induction over the cells in $\MFWMode$.
Starting with the given polygon, $P$, and the given start cell, $s$, we
can define a sequence $P_{k,i}$ of polygons $(P_{0, 0} := P)$ with start cells 
$s_{k, i}$ as follows:
$P_{k, i+1}$ arises from $P_{k, i}$ by removing the start cell
$s_{k, i}$ and all cells that are reserved in the first step
in $P_{k, i}$ (\IE, every cell $c$ with $\varphi(c) = s_{k, i}$).
The start cell $s_{k, i+1}$ in the new polygon $P_{k, i+1}$
is the cell that the robot enters with its first step in $P_{k, i}$,
see \reffig{figCells/knabbernbeispiel}. The reserved cells in this and all
following figures are shown shaded; $\triangle$ denotes the
start cell.

There is nothing to consider when the strategy enters the backward mode,
because we remove all cells that are explored in the backward mode
together with the forward cells.
But what happens, if a division cell occurs; that is, the strategy switches
from the backward mode to the forward mode?

\begin{lem}{zerfall}
If one of the cells reserved in the first step in $P_{k, i}$
is a division cell, then $P_{k, i}$ is split by
removing $s_{k,i}$ and all cells that are reserved in this step (\IE,
all cells $c$ with $\varphi(c)=s_{k,i}$)
into two or more not-connected components.
\end{lem}
\begin{proof}
Consider two components, $P_1$ and $P_2$, that are connected in $P$
by some cells $\MX \subset \MBKMode$. 
Let $c_j$ be the first cell in $\MX$ that is discovered on the return path,
thus, $P_2$ is entered via $c_j$. 
In our successive decomposition, $\varphi(c_j)$ is the start cell
of a polygon $P_{k, i}$. In $P_{k, i}$ all cells explored before 
$\varphi(c_j)$ are already removed. 
If there would be another connection, $c_\ell$, between $P_1$ and $P_2$ 
at this time, $c_\ell$ would be discovered before
$c_j$ on the return path and $P_2$ would have been entered via $c_\ell$ in 
contradiction to our assumption that $P_2$ is entered via $c_j$.
Thus, $c_j$ must be the last cell that connects $P_1$ and $P_2$.
\end{proof}

If one or more of the cells to be removed are division cells, 
the polygon $P_{k, i}$ is divided into subpolygons 
$P_{k+1, 0}, P_{k+2, 0}, \ldots$ that are analyzed separately,
see \reffig{figCells/splitbeispiel}. See 
\reffig{figCells/knabbernbeispiel2} for
a more comprehensive example for the successive
decomposition of a polygon.

\pstexfig{Handling of division cells.}{figCells/splitbeispiel}

\bigskip\noindent
Now, we are able give a first bound for the number of steps that
\Strategyname\ uses to explore a polygon.

\begin{lem}{schranke1}
Let us assume that the cell behind and right hand to 
the robot's position\footnote{In other words, the cell southeast to 
the robot's current position if the current direction is {\em north}.}
is blocked.
The number of steps, $S$, used to explore a polygon with $C$ cells,
$E$ edges and $H$ holes, is bounded by
$$ S \leq\ C + \Ehalbe + H + 2L - 3,$$
where $L$ denotes the number of the robot's left turns.
\end{lem}

\begin{proof*}
We observe the differences in the number of steps, cells, edges
and holes between $P_{k, i}$ and $P_{k, i+1}$, and assume by
induction that our upper bound for the length of the exploration tour holds 
for $P_{k, i+1}$ and for the separated subpolygons $P_{k+j, 0}$.
Therefore, we have
to show that the limit is still not exceeded if we add the removed
cells and merge the subpolygons.
We want to show
that the following inequation is satisfied in every step:
$$ S \leq\ C + \Ehalbe + H + 2L - 3.$$

Let $G$ denote the ``profit'' made by \Strategyname; that is, the difference
between the actual number of steps and the upper bound. With $G$,
the preceding inequation is equivalent to
$$C + \Ehalbe + H + 2L - G - S - 3 = 0.$$

\pstexfig{Decomposing a step to the right into several forward steps.}
{figCells/decompright}

We have to consider three main cases: the division cells, the cells 
contained in $\MFWMode_L$ (left turns), and those contained in 
$\MFWMode_F \cup \MFWMode_R$ (forward steps and right turns).
There is no need to consider right turns explicitly, 
because steps to the right can be handled as a sequence
of forward steps, see \reffig{figCells/decompright}. 
The successive decomposition ends
with one single cell ($C=1, E=4, H=0$) for which 
$S=0 =  C + \Ehalbe + H + 2L - 3$ holds (Induction base).

\pagebreak 

\begin{description}
\item[Division cells]~\\
If one of the cells that are reserved in the first step in $P_{k, i}$
is a division cell, $P_{k, i}$ is split into two polygons
$P_{k+1, 0}$ ($P_1$ for short) and $P_{k+2, 0}$ ($P_2$ for short), 
see \reflem{zerfall}. We assume by induction that our upper bound is
achieved in both polygons: 
$$C_i + \Ehalbe_i + H_i + 2L_i - G_i - S_i - 3 = 0,\quad i\in\{1,2\}.$$

\noindent
For the merge of $P_1$ and $P_2$ into one polygon $P$, we can
state the following:
\begin{samepage}
\begin{eqnarray*}
S &=& S_1 + S_2 + \Delta S \\
E &=& E_1 + E_2 + \Delta E \\
C &=& C_1 + C_2 \\
H &=& H_1 + H_2 \\
L &=& L_1 + L_2 \\
G &=& G_1 + G_2 + G_S, 
\end{eqnarray*}
where $G_S$ denotes the profit made by merging the polygons.
\end{samepage}

We want to show that our bound is achieved in $P$:
\begin{eqnarray*}
\lefteqn{C + \Ehalbe + H +2L - G - S - 3} \\
&=&
C_1 + C_2 + {1 \over 2} ( E_1 + E_2 + \Delta E ) + H_1 + H_2+ 2L_1 \\
& &\mbox{} + 2L_2 - G_1 - G_2 - G_S - S_1 - S_2 - \Delta S - 3 \\
&=& \underbrace{C_1 + \Ehalbe_1 + H_1 + 2L_1 - G_1 - S_1 - 3}_{=0}  \\
&& \mbox{}+\underbrace{C_2 + \Ehalbe_2 + H_2 + 2L_2 - G_2 - S_2 - 3}_{=0}\\
&& \mbox{}+{1 \over 2} \Delta E - G_S - \Delta S + 3 \\
&=& 0 \\
\Longleftrightarrow \;\; G_S &=& {1 \over 2} \Delta E - \Delta S + 3 
\end{eqnarray*}
Thus, for every polygon split we have to observe $\Delta E$ and
$\Delta S$. If $G_S$ is positive, we gain some steps by
merging the polygons, if $G_S$ is negative, the merge incurs some costs.

The different configurations for a polygon split can be assigned to
several classes, depending on the number of 
common edges between $P_1$ and $P_2$ and the way, the robot returns
from $P_2$ to $P_1$---more precisely the distance between the
step from $P_1$ to $P_2$ and the step from $P_2$ to $P_1$, 
compare for example \reffig{figCells/splitfaelle}(3a) and (3b).
\reffig{figCells/splitfaelle}(i) and (ii) show some instances of one class.
Because the actual values for $\Delta E$ and $\Delta S$ depend
on only these two parameters, we do not list every possible polygon split 
in \reffig{figCells/splitfaelle}, but some instances of every class.

The balances of the polygon splits are shown in table \ref{tab-splitfaelle}. 
Two cases have a negative balance, so we need one more argument to show
that these cases do not incur any costs. Observe that after splitting
$P_{k,i}$ the polygon $P_1$ starts with a cell from $\MFWMode_L$
in the cases (1b)--(5). Removing this block of 4 cells, we gain
$+2$, see the first line of \reffig{figCells/ltsteps1}.
This covers the costs for the polygon split in the cases
(4a) and (5a).

\item[Forward steps]~\\
We have to consider several cases of forward steps, 
see \reffig{figCells/fwdsteps}. 
The table lists the differences
in the number of steps ($\Delta S$), cells ($\Delta C$), edges
($\Delta E$) and holes ($\Delta H$) if we {\em add} the considered cells
to $P_{k,i+1}$.
The last column shows 
$G = \Delta C + {1\over 2} \Delta E + \Delta H - \Delta S$.
After removing the observed block of cells, the remaining polygon must still 
be connected; otherwise, we would have to consider a polygon split first.
Thus, there are some cases, in which $\Delta H$ must be greater than zero.

It turns out that all cases
have a positive balance, except those that violate the 
assumption in \reflem{schranke1}
that the cell behind and right hand to the robot's position is blocked,
see the cases marked with (*) in \reffig{figCells/fwdsteps}.
Notice that the configurations shown in 
\reffig{figCells/NoFwdSteps} are left turns instead of forward steps!

To show that we have analyzed all possible 
cell configurations for a step forward, we use the following
observations.
When the robot makes a step forward, we know the following: Both behind
and left hand to the robot are walls (otherwise it would have turned left),
in front of the robot is no wall (otherwise it could not make a step
forward). Right hand to the robot may or may not be a wall. In the
latter case, we have three edges of interest
that may or may not be walls, yielding $2^3=8$ cases, which can be 
easily enumerated.
Two special cases occur by taking into consideration that
the robot may enter the observed block of cells from the upper cell or
from the right cell, see for example \reffig{figCells/fwdsteps}(5a) and (5b).

\pstexfig{Configurations that are steps to the left instead of forward steps.}
{figCells/NoFwdSteps}

\item[Steps to the left]~\\
Possible cases of left turns are shown in the figures
\ref{figCells/ltsteps1-fig}--\ref{figCells/ltsteps4-fig}. 
As in the previous case, the last column shows the balance. 
Again, we observe that all cases with a negative balance of $-3$ are
a violation of our assumption that the cell behind and right hand to 
the robot's position is blocked.
The negative balances of $-2$ are 
compensated by the addend $2L$ in our bound.

\pstexfig{Another class of left turns.}
{figCells/ltbeispiel}

When counting the number of steps, we assumed that the robot enters 
the block of cells from the same direction as it left the block 
(from the left as shown in the figures). The robot may enter
the block from the upper cell as shown in \reffig{figCells/ltbeispiel},
but this would only increase the balance.

The completeness of the cases can be shown with the same argument
as in the previous case: We know that there is a wall behind the robot and left hand
to the robot is no wall. Examining a block of four cells for a left turn,
we have six edges that may or may not be walls, and we have three cells
that may or may not be holes, yielding corresponding configurations
of edges.\proofendbox
\end{description}
\end{proof*}

\begin{table}[h]
\begin{center}
\begin{tabular}{cccc}
& $\Delta E$ & $\Delta S$ & $G_S$ \\
(1)   &   -2 &        2   &    0  \\
(2)   &   -4 &        0   &    1  \\
(3a)  &   -6 &        0   &    0  \\
(3b)  &   -6 &       -2   &    2  \\
(4a)  &   -8 &        0   &   -1  \\
(4b)  &   -8 &       -2   &    1  \\
(4c)  &   -8 &       -4   &    3  \\
(5a)  &  -10 &        0   &   -2  \\
(5b)  &  -10 &       -2   &    0  \\
\end{tabular}
\end{center}
\caption{Balances of polygon splits\label{tab-splitfaelle}}
\end{table}

\clearpage

\pstexfig{Some possible cases of polygon splits.}{figCells/splitfaelle}

\pstexfig{Cell configurations for forward steps.}{figCells/fwdsteps}

\psfig{Possible configurations for steps to the left (1).}{figCells/ltsteps1}

\psfig{Possible configurations for steps to the left (2).}{figCells/ltsteps2}

\psfig{Possible configurations for steps to the left (3).}{figCells/ltsteps3}

\psfig{Possible configurations for steps to the left (4).}{figCells/ltsteps4}

\clearpage

\noindent
Next, we want to discard the assumption we made in \reflem{schranke1}.
\begin{lem}{firststep}
The assumption that the cell behind and right hand to 
the robot's position is blocked 
can be violated only in the robot's initial position;
that is, in $P_{0,0}$ and not in any $P_{k,i}$ with $k+i>0$.
\end{lem}
\begin{proof*}
\begin{smallfig}{figCells/firststep}
Consider a robot located in some cell $c$ with no wall behind and right hand to
its position, and \WLOG\ $\VAR{dir}=$'north'.
If this is not the robot's initial position, but the position $s_{k,i}$ of
a polygon $P_{k,i}$ occurring in the successive decomposition of the polygon,
the robot must have entered the cell $c$ from the cell $c'$ below $c$ in the
polygon $P_{k,i-1}$. If the cell $c''$ right to $c'$ is not a hole in 
$P_{k,i-1}$, $c''$ would be a reserved cell, and, thus, it would be 
removed together with $c'$ in the step from $P_{k,i-1}$ to $P_{k,i}$. 
Consequently, when the robot has reached $c$, there would be a hole 
behind and right hand to its current position.\proofendbox
\end{smallfig}
\end{proof*}

\begin{lem}{schranke2}
The number of steps, $S$, used to explore a polygon with $C$ cells,
$E$ edges and $H$ holes, is bounded by
$$ S \leq\ C + \Ehalbe + H + 2L - 2,$$
where $L$ denotes the number of the robot's left turns.
\end{lem}

\begin{proof}
\reflem{firststep} shows that the assumption in \reflem{schranke1}
can be violated only in 
the robot's initial position. On the other hand, we have seen in the proof of 
\reflem{schranke1} that all cases that violate the assumption 
incur the costs of 
just one additional step.
\end{proof}

\pstexfig{Corridors of odd width.}{figCells/oddcorridor}

Our bound still depends on the number of left turns the robot makes while
exploring the polygon. To give a bound that does not depend on the
strategy, we introduce another property of grid polygons, 
a measure to distinguish rather
flat polygons from winded polygons; let us call it the
{\em sinuosity} of $P$.
Our motivation for introducing this property is
the following observation: The robot may walk 
$n+1$ times through a corridor of width $n$, $n$ odd, 
see \reffig{figCells/oddcorridor}(i). 
The costs for this extra walk are covered by the corridor walls;
more precisely, for every double visit we charge {\em two} polygon edges
and get the addend ${1\over 2}E$ in the upper bound. 
If there is a left turn in the corridor, there are not enough boundary edges
for balancing the extra walk. 
\reffig{figCells/oddcorridor}(ii) shows a corridor
of width 3 with a left turn. The steps shown with dashed lines cannot be
assigned to edges, so we have to count the edges shown with dashed lines
to balance the number of steps.
We define two types of sinuosities, the clockwise and the counterclockwise
sinuosity. Because \Strategyname\ follows the left-hand rule, 
our strategy depends on the clockwise sinuosity. A similar strategy that
follows the right-hand rule would depend on the counterclockwise
sinuosity.

\pstexfig{Contributions to $\WCW$ by (i) the outer boundary, (ii) 
inner boundaries.}
{figCells/WcwInnerOuter}

\pstexfig{Reflex vertices $p_i$ and the corresponding squares $Q_i$. 
$\WCW=q_1'+q_3'=8, \WCCW=q_4'=2$.}
{figCells/convexvertex}

\begin{defi}{winding}
Let the \DEFI{clockwise} \DEFI{sinuosity}, $\WCW$, and the 
\DEFI{counterclockwise sinuosity}, $\WCCW$, 
of a grid polygon $P$ be defined as follows:
We trace the boundary of $P$---the outer boundary clockwise, the
boundaries of the holes inside $P$ counterclockwise---, and
consider every pair, $p_i$ and $p_{i+1}$, of consecutive reflex vertices,
see \reffig{figCells/WcwInnerOuter}.

We trace the angular bisector between the two edges incident to $p_i$
inside $P$ until it hits the boundary of $P$. The resulting line
segment defines the diagonal of a square, $Q_i$,
see \reffig{figCells/convexvertex}.%
\footnote{We can construct $Q_i$ by ``blowing up'' a square around the 
cell in $P$ that touches the boundary of $P$ in $p_i$ until
the corner of $Q_i$ opposite to $p_i$ hits the polygon's boundary.}
Let $q_i$ be the width of $Q_i$, analogously with $q_{i+1}$.

Because the robot needs some further steps only in odd corridors, we count
only odd squares:
$$q_i':=\cases{q_i-1, & if $q_i$ is odd \cr 0, & if $q_i$ is even}\,.$$
The need for additional edges may not only be caused by reflex vertices,
but also by the start cell, see \reffig{figCells/WcwExa}(ii).
Thus, we consider the squares $Q_{\rmsub{s}{cw}}$ 
and $Q_{\rmsub{s}{ccw}}$ from the start cell in clockwise and
counterclockwise direction, respectively. Let $q_{\rmsub{s}{cw}}'$ 
and $q_{\rmsub{s}{ccw}}'$ be defined analogously to $q_i'$.
Now, we define the clockwise sinuosity $\WCW$ 
and the counterclockwise sinuosity $\WCCW$ 
as 
$$\WCW:=q_{\rmsub{s}{ccw}}'+\sum_{i\geq 1} q_{2i-1}', 
\quad \mbox{\rm and} \quad
\WCCW:=q_{\rmsub{s}{cw}}'+\sum_{i\geq 1} q_{2i}'.$$
\end{defi}

\reffig{figCells/WcwExa} shows two examples for the definition of $\WCW$.
Note that in (i) only one reflex vertex contributes to $\WCW$, and
every edge we count here is needed.

\pstexfig{Examples for the definition of $\WCW$: (i) A polygon with
$C=193, \frac{E}2=78, H=3, \WCW=6, S=284$ 
(the bound for $S$ is exactly achieved),
(ii) the start cell contributes to $\WCW$, too ($C=46, \frac{E}2=23,
H=2, \WCW=2, S=74$).}
{figCells/WcwExa}


\pagebreak 

\noindent
With the definition of $\WCW$, we can give our final result:

\begin{theo}{schrankefinal}
Let $P$ be a grid polygon with $C$ cells, $E$ edges, $H$ holes, and
clockwise sinuosity $\WCW$. 
\Strategyname\ explores $P$ using no more than
$$ S \leq\ C + \Ehalbe + \WCW + 3H - 2$$
steps. This bound is tight.
\end{theo}

\pstexfig{Left turn followed by (i) a right turn and (ii) a reduction.}
{figCells/narrowrt}

\begin{proof}
We need some global arguments to charge the costs for a left turn
to properties of $P$. So let us examine, which configurations 
may follow a left turn (after some forward steps):

\begin{itemize}
\item A right turn follows the left turn, see \reffig{figCells/narrowrt}(i).
  We gain $+2$ steps per right turn, so the possible costs of $-2$ for
  this left turn are covered.

\item An obstacle follows the left turn. We can charge the obstacle
  with the costs for the left turn and get a factor of 3 for the
  number of obstacles. Every obstacle is charged at most once, because
  when the successive decomposition reaches the obstacle for the first
  time, the obstacle disappears; that is, the hole merges with the outer
  boundary.

\item A reduction follows the left turn, see \reffig{figCells/narrowrt}(ii).
  Later in this section, we show that a reduction covers the costs of a left turn.

\item Another left turn follows the observed left turn. In this case, there
  is no other property of $P$ to be charged with this costs but the
  sinuosity $\WCW$, this follows directly from the definition of $\WCW$.
\end{itemize}

\pstexfig{(i) A reduction follows a left turn, (ii) the left turn causes a 
polygon split, (iii) no forward steps between the left turn and the
reduction ($d=0)$.}
{figCells/Reduction1}

In the case of a reduction following a left turn
we observe the number, $d$, of forward steps between the left turn
and the reduction, as well as the cell marked with $b$, 
and---if $d\geq 1$---the 
cell marked with $a$ in \reffig{figCells/Reduction1}(i).
If $b$ is blocked, $b$ is either part of an
obstacle inside the polygon or it is outside the polygon. In the
first case, we charge the obstacle with the costs of the left turn
as described earlier. In the second case we have a polygon split that
leaves us with a left turn that incurs no costs, see
\reffig{figCells/Reduction1}(ii) and the first line of \reffig{figCells/ltsteps1}.
The same holds for the cell marked with $a$ if there is at least
one forward step between the left turn and the reduction (\IE, $d \geq 1$).
Therefore, we assume that $a$ and $b$ are free cells in the following.

If the reduction follows immediately after the left turn ($d=0$), see
\reffig{figCells/Reduction1}(iii), we have one of the left turns shown 
in \reffig{figCells/ltsteps2}, \reffig{figCells/ltsteps4} or the lower half of 
\reffig{figCells/ltsteps3}. In any of these cases we have either a positive
balance or we meet an obstacle ($\Delta H>0$) and charge the costs for
the left turn to the obstacle as earlier.

\pstexfig{If $a$ and $b$ are free cells we gain 2. 
(i) $d=1$, (ii) $d\geq 1$.}
{figCells/Reduction2}

If there is one forward step between the left turn and the reduction
($d=1$), the robot enters the $2\times 2$ block of cells of the left turn
not from the same side as it left it, because $a$ and $b$ are polygon cells.
In \reffig{figCells/Reduction2}(i) the robot leaves the block to 
the left but enters it from above.
This situation is described in \reffig{figCells/ltbeispiel} on
\refpagefig{figCells/ltbeispiel} and reduces the costs for the left turn by 2, 
so the balance is either zero or positive.
If there is more than one forward step ($d>1$), we have either the
same situation as in the preceding case, or the reduction from a corridor
of width $\geq 3$ to a corridor of width $\leq 2$ shifts to the left,
see \reffig{figCells/Reduction2}(ii), and eventually we reach a forward step
as shown in \reffig{figCells/fwdsteps}(5b) that gains $+2$ and covers the
costs for the left turn.

Altogether, we are able to charge the costs for every left turn to
other properties, which proves our bound.

\reffig{figCells/WcwExa} and
\reffig{figCells/CellExplTight} show 
nontrivial examples (\IE, $H\neq 0$ and $\WCW \neq 0$) for 
polygons in which the bound is exactly achieved.
\end{proof}

\pstexfig{Polygon with $C=34, \frac{E}2=17, H=1, \WCW=2, 
S=54=C+\frac12E+3H+\WCW-2$.}{figCells/CellExplTight}

%% file: remarks.tex
\sect{Summary}{cellconcl}

We considered the exploration of grid polygons.
For simple polygons we have shown a lower bound of $\frac76$ and 
presented a strategy, \smartDFS, that explores simple polygons with 
$C$ cells and $E$ edges using no more than $C+\frac12 E-3$ steps
from cell to cell. 
Using this upper bound, we were able to show that \smartDFS\ is
in fact $\frac43$-competitive, leaving a gap of only $\frac16$ between
the upper and the lower bound.

On the other hand, the competitive complexity 
for the exploration of grid polygons with holes is $2$. A simple
DFS exploration already achieves this competitive factor, but, nevertheless,
DFS is not the best possible strategy, because is it not necessary to
visit {\em each} cell twice. Therefore, we developed the
strategy \Strategyname\ that takes advantage of wider areas in the
polygon, and thus corrects the weakness in the DFS strategy.

Interesting open questions are, how SmartDFS and CellExplore can be
generalized to higher dimensions and other cell types than squares, e.g.\
triangular or hexagonal cells.